\documentclass[12pt]{article}
\usepackage[centertags]{amsmath}
\usepackage{amssymb,amsfonts,latexsym}
\usepackage{amsthm}
\usepackage{times}
\usepackage{array}
\usepackage{graphicx}
\usepackage{color}
\usepackage{xcolor}
\newtheorem{Thm}{Theorem}[section]

\newtheorem{Def}{Definition}[section]
\newtheorem{Lem}[Thm]{Lemma}
\newtheorem{Prop}{Proposition}[section]
\newtheorem{Rem}{Remark}[section]

\newcommand{\hslashslash}{%
  \raisebox{.9ex}{%
    \scalebox{.7}{%
      \rotatebox[origin=c]{18}{$-$}%
      }%
  }%
}
\newcommand{\dslash}{%
  {%
   \vphantom{d}%
   \ooalign{\kern.05em\smash{\hslashslash}\hidewidth\cr$d$\cr}%
   \kern.05em
  }%
}

\def\={\hspace{-3mm}&=&\hspace{-3mm}}

\title{The fractional in time Schr\"{o}dinger equation with a Hartree perturbation}
\author{ Humberto Prado$^a$\footnote{E-mail: humberto.prado@usach.cl} \, and Jos\'e Ram\'irez$^b$\footnote{E-mail: jose.ramirezm@usach.cl} \\
\small{$^{a,b}$ Departamento de Matem\'atica y Ciencia de la
Computaci\'on,} \\
\small{ Universidad de Santiago de Chile }\\
\small{Casilla 307 Correo 2, Santiago-Chile }}

\begin{document}
\date{}

\maketitle


\begin{abstract}
The aim of this work is to show existence, uniqueness and regularity properties of nonlinear fractional Schr\"{o}dinger equation (\ref{ec-princ}) with fractional time derivative  of order $\alpha\in (0,1)$ and with a Hartree-type of  nonlinear term.

\end{abstract}


\section{Introduction}\label{sec1}
 The fractional in time linear Schr\"{o}dinger equation  has been studied in \cite{Gorka2017} in which the abstract  fractional evolution equation has been investigated in the general setting of Hilbert spaces. We point out that the fractional in time Schr\"{o}dinger equation has applications in the context of quantum fractional mechanics; see \cite{L} and the references in there.\\
The main results of this research are strongly motivated by recent investigation on non-linear semi-relativistic Schr\"{o}dinger equations
 e.g., \cite{Hajaiej2013},\cite{Kwon2015},\cite{Herr2015};
This class of equations have interesting applications for a large systems of self-interactions, and the effective description of pseudo-relativistic boson stars via a Coulomb law; see e.g., \cite{Lenzmann2004},\cite{Jonsson2007},\cite{Lewin2011} and the references given there. Nevertheless, to the best of our knowledge the analogous problem with a fractional time derivative has not been  investigated so far. Thus,  our main purpose in this paper is to study the following  time fractional evolution non-linear problem with time fractional derivative in the Caputo sense,
\begin{equation}\label{ec-princ}
\begin{cases}
i^{\alpha}D_t^{\alpha} u(t,x)=(-\Delta)^{\beta/2}u(t,x)+\lambda J_t^{1-\alpha} K_{\gamma}(|u|^2)(x)u(t,x), &\ (t,x)\in[0,T]\times\mathbb{R}^n, \\
\quad\quad\ \ \ u(0,x)=u_0(x),\ u_0\in H^{\beta}(\mathbb{R}^n),     &
\end{cases}
\end{equation}
in which $K_{\gamma}(|u|^2)\,u$ is the Hartree potential. We assume that $\beta>0$, $\alpha\in(0,1),$ $i^{\alpha}=e^{i\alpha\pi/2},$\, and $\lambda\in\mathbb{R}\backslash\{0\}.$ The nonlocal operator $(-\Delta)^{\beta/2}$ is defined as a pseudo-differential operator with the symbol $|\xi|^{\beta}$ on $\mathbb{R}^n$, $J_t^{1-\alpha}$ is the fractional integral in the Riemann-Liouville sense of order $1-\alpha,$ and the non-linear term is defined defined for each $u\in L^p:=L^p(\mathbb{R}^n),$ by the convolution operator,
\begin{equation}\label{htree}
    K_{\gamma}(u)(x)=\int_{\mathbb{R}^n}\frac{\psi(x-y)}{|x-y|^{\gamma}}u(y)dy,
\end{equation}
where $\psi$ is assumed to be nonnegative and bounded, and  $\gamma\in(0,n).$ Henceforth we denote $\psi_{\gamma}:=\displaystyle\frac{\psi}{|\cdot|^{\gamma}},$ and thus
$$K_{\gamma}(u)(x)=(\psi_{\gamma}\ast u)(x);$$ see (\ref{des-Hardy}) and  (\ref{1Leib-16}) below for definitions and further properties.

\section{Preliminaries}

In this section we establish the basic notations, and the technical results which will be used thereafter.

\subsection{The fractional derivative of Caputo}
Hereafter, we denote
$$g_{\alpha}(t)=\displaystyle\frac{t^{\alpha-1}}{\Gamma(\alpha)},\quad \mbox{for}\quad \alpha>0,\quad  t>0.$$
Then we define the Riemmann Liouville integral as
$$J_t^{\alpha}u(t)=\frac{1}{\Gamma(\alpha)}\displaystyle\int_0^t(t-s)^{\alpha-1}u(s)ds,$$ for a given locally integrable function  $u$  defined on  the half line $\mathbb{R}_+:=[0,\infty)$ and taking values on a Banach space $X.$ Henceforth we use the notation,
\begin{align*}
J_t^{\alpha}u(t)=(g_{\alpha}*u)(t),\quad t>0.
\end{align*}

Then the following property holds: $J_t^{\alpha+\gamma}u(t)=J_t^{\alpha}J_t^{\gamma}u(t),$ for $\alpha,\gamma>0,$ in which $u$ is suitable enough.

We shall consider  the following definition of the fractional derivative of order $\alpha\in (0,1).$  Assume  that $u\in C(\mathbb{R}_+;X)$ and that the convolution $g_{1-\alpha}*u$\,\, belongs to\,\, $C^1((0,\infty);X).$ Then the Caputo fractional derivative of order $\alpha\in (0,1),$ can be interpreted as
$$
  D_t^{\alpha} u(t)=\frac{d}{dt}(g_{1-\alpha}*u)(t)-u(0)g_{1-\alpha}(t)=\frac{1}{\Gamma(1-\alpha)}\left[\frac{d}{dt}\left(\int_0^t(t-s)^{-\alpha}u(s)ds\right)-\frac{u(0)}{t^{\alpha}}\right].
$$
Furthermore if $u\in AC(\mathbb{R}_+;X),$ in which $AC(\mathbb{R}_+;X)$ is the space of absolutely continuous functions on $\mathbb{R}_+,$ then we can also realize the Caputo derivative as
\begin{equation}\label{cap1}
D_t^{\alpha}u(t)=J_t^{1-\alpha}u'(t)\quad\quad\mbox{for}\quad\quad \alpha\in(0,1);
\end{equation}
see e.g., \cite{Podlubny1999} for further properties and definitions.

Henceforth we shall denote the Caputo derivative by $D_t^{\alpha} u(t)$.

\begin{Rem}\label{rmk-1}
We recall the Mittag-Leffler function (see e.g., \cite{Podlubny1999}),
\begin{equation*}
E_{\alpha,\eta}(z):=\sum_{n=0}^{\infty}\frac{z^n}{\Gamma(\alpha n+\eta)},\ \alpha>0, \eta\in\mathbb{C},z\in\mathbb{C}.
\end{equation*}
The function $E_{\alpha,\eta}$ is an entire function of $z.$ We denote $E_{\alpha,1}(z)=E_{\alpha}(z),$ for $\alpha>0, z\in\mathbb{C}.$
Next we record the following estimates satisfied by the Mittag-Leffler function \cite[Theorem 1.5, page 35]{Podlubny1999} (see also \cite[Lemma 2.2]{Gorka2017}). For $\alpha\in(0,1),\beta>0$ there exists a positive constant $M_0,$ such that
\begin{equation}\label{Est-Mit}
|E_{\alpha}((-it)^{\alpha}|\xi|^{\beta})|\leq M_0,\quad t>0,\xi\in\mathbb{R}^n.
\end{equation}
\end{Rem}

\subsection{Fractional Sobolev spaces}
For $\beta>0$ and $p\geq 1$, we define the fractional Sobolev space,
\begin{equation}\label{norm-sp}
H^{\beta,p}=H^{\beta,p}(\mathbb{R}^n)=\{u\in L^p:\mathcal{F}^{-1}[(1+|\xi|^2)^{\beta/2}\hat{u}(\xi)]\in L^p\},
\end{equation}
endowed with the norm
\begin{equation*}
\|u\|_{H^{\beta,p}}=\|\mathcal{F}^{-1}[(1+|\xi|^2)^{\beta/2}\hat{u}(\xi)]\|_{L^p},
\end{equation*}
in which $\hat{u}:=\mathcal{F}(u)$ stands for the Fourier transform  of $u.$ Then $H^{\beta,p}$ is Banach space endowed  with the norm $\|\cdot\|_{H^{\beta,p}},$
see \cite{Bergh1976,Taylor2010}.
In particular, we shall denote $H^{\beta,2}$ as $H^{\beta}.$  We define
\begin{equation*}
(I-\Delta)^{\beta/2}(u)=\mathcal{F}^{-1}[(1+|\xi|^2)^{\beta/2}\hat{u}(\xi)],\  \mbox{for}\ u\in H^{\beta,p}.
\end{equation*}
Therefore we denote, $\|u\|_{H^{\beta,p}}=\displaystyle\|(I-\Delta)^{\beta/2}(u)\|_{L^p}.$

\begin{Def}\label{def-lapla}(see \cite{Mateusz2017})
Let $\beta>0$ be fixed. Then we define the fractional laplacian $(-\Delta)^{\beta/2},$ as follows
\begin{equation*}
    (-\Delta)^{\beta/2}u=\mathcal{F}^{-1}[|\xi|^{\beta}\hat{u}(\xi)],\quad u\in \mathcal{D}((-\Delta)^{\beta/2}),\xi\in\mathbb{R}^n,
\end{equation*}
on the domain
\begin{equation}\label{dom-m}
\mathcal{D}((-\Delta)^{\beta/2})=\left\{u\in L^2:\int_{\mathbb{R}^n}|\xi|^{2\beta}|\hat{u}(\xi)|^2d\xi<\infty\right\}.
\end{equation}
\end{Def}
Let $\beta>0.$ Then from the Definition \ref{def-lapla} follows that
\begin{equation*}
\mathcal{D}((-\Delta)^{\beta/2})=H^{\beta}.
\end{equation*}

Let  ${\mathcal{S}}_0:={\mathcal{S}}_0(\mathbb R^n)$ be  the space of all those $u$ in the Schwartz space ${\mathcal{S}}:={\mathcal{S}}(\mathbb R^n)$  such that its Fourier transform $\widehat{u}$ vanishes on a neighborhood of the origin. Then we define for $0<\beta<n/2$ and $1\leq p<\infty$ the homogeneous Sobolev space $\dot{H}^{\beta,p}$ as the completion of ${\mathcal{S}}_0$ with the norm
\begin{equation}\label{Homspace}
\|u\|_{\dot{H}^{\beta,p}}=\|\mathcal{F}^{-1}[|\xi|^{\beta}\hat{u}(\xi)]\|_{L^p}.
\end{equation}

Then, $\dot{H}^{\beta,p}$ is a Banach space contained in the space of tempered distributions ${\mathcal{S}}'$; see \cite{Bergh1976}. In particular, we shall denote $\dot{H}^{\beta,2}$ as $\dot{H}^{\beta}.$


{\rm{We state the following known facts that will be needed in the forthcoming sections.}}

\begin{Rem}\label{rmk-2}\,\rm{(Sobolev's embedding, see \cite{Taylor2010})}.
\begin{itemize}
\item[(i)] Let $\gamma_1\leq \gamma_2,$ $1\leq p\leq\infty.$ Then,
\begin{equation}
    H^{\gamma_2,p}\hookrightarrow H^{\gamma_1,p}.
\end{equation}

\item[(ii)] If $1\leq p<\infty$ and $0<\gamma<\frac{n}{p},$ then,
\begin{equation}
    H^{\gamma,p}\hookrightarrow L^{np/n-\gamma p},
\end{equation}
in particular if $p=2,$\, and\, $0<\gamma/2<n/2$ then,
\begin{equation}\label{h1}
    H^{\gamma/2}\hookrightarrow L^{2n/n-\gamma }.
\end{equation}
\end{itemize}
Furthermore, by (i) for $\gamma_1=\gamma/2,\gamma_2=\beta,p=2$ together with (ii) for $0<\gamma/2<n/2$ we obtain  the embedding,
\begin{equation}\label{h2}
 H^{\beta} \hookrightarrow  H^{\gamma/2}\hookrightarrow L^{2n/n-\gamma }.
\end{equation}
\end{Rem}

\begin{Rem}\label{rmk-3}
\begin{itemize}
\item[(i)]{\rm{Let $0<\beta<1.$ Then the norm $\|\cdot\|_{H^{\beta}}$ of $H^{\beta}$ is equivalent to the graph norm of the fractional Laplacian operator $(-\Delta)^{\beta/2}$ on $ L^{2},$ that is,}}
\begin{equation*}
\|u\|_{(-\Delta)^{\beta/2}}=\|u\|_{L^2}+\|(-\Delta)^{\beta/2}u\|_{L^2},\,\,\,\,\,\,\,\, u\in H^{\beta}.
\end{equation*}

\item[(ii)] {\rm{(Hardy inequality, see \cite{Tao2006}). Let $0\leq s<n/2,$ $u\in \dot{H}^s.$ Then there exists non-negative constant $C$ such that,}}
\begin{equation*}
\left\|\frac{u}{|\cdot|^s}\right\|_{L^2}\leq C\|u\|_{\dot{H}^s}.
\end{equation*}

\item[(iii)] {\rm{(Fractional Leibniz rule, see \cite{Gr2014}). Let $\sigma>0,$ $1<r<\infty$ and $1<p_i,q_i\leq\infty$ and suppose that $\displaystyle\frac{1}{r}=\frac{1}{p_i}+\frac{1}{q_i},$ for $i=1,2$. Then there exists a positive constant $C$ such that for each $u,v\in \mathcal{S},$}}
\begin{equation}\label{Leib-15}
 \displaystyle\|(-\Delta)^{\sigma/2}(uv)\|_{L^r}\leq C\left(\|(-\Delta)^{\sigma/2}u\|_{L^{p_1}}\|v\|_{L^{q_1}}+\|u\|_{L^{p_2}}\|(-\Delta)^{\sigma/2}v\|_{L^{q_2}}\right),
\end{equation}
\begin{equation}\label{1Leib-15}
 \displaystyle\|(I-\Delta)^{\sigma/2}(uv)\|_{L^r}
 \leq C\left(\|(I-\Delta)^{\sigma/2}u\|_{L^{p_1}}\|v\|_{L^{q_1}}
 +\|u\|_{L^{p_2}}\|(I-\Delta)^{\sigma/2}v\|_{L^{q_2}}\right).
\end{equation}

\item [(iv)]\label{des-hlitt} {\rm{(Hardy-Littlewood-Sobolev inequality, see \cite{Stein1993}). Let $\gamma\in (0,n),$ $1<p<q<\infty$, $u\in L^p.$ Then under the assumption that
\begin{equation*}
\frac{1}{q}=\frac{1}{p}-\frac{(n-\gamma)}{n},
\end{equation*}
there exists a positive constant $C$ such that}}
\begin{equation}\label{K-q}
\left\|u \ast \frac{1}{|\cdot|^{\gamma}}\right\|_{L^q}\leq C\|u\|_{L^p},
\end{equation}

\item [(v)] {\rm{Given  $u\in L^p$ and assuming  that $\gamma, p$ and $q $ satisfy the same conditions as in ($iv$) above. Then we obtain the following direct consequence of (\ref{1Leib-15})}}
\begin{equation}\label{1Leib-16}
\|K_{\gamma}(u)\|_{L^q}
\leq \|\psi\|_{\infty}\left\|u \ast \frac{1}{|\cdot|^{\gamma}}\right\|_{L^q}
\leq C \|\psi\|_{\infty}\|u\|_{L^p},
\end{equation}
where $C$ is a positive constant.
\end{itemize}
\end{Rem}

\begin{Lem}\label{des-gam} Let $\gamma\in(0,n),$ $u\in \dot{H}^{\gamma/2}.$ Then there exists a positive constant $C$ such that
\begin{equation}\label{des-Hardy}
\sup_{y\in \mathbb{R}^n}\int_{\mathbb{R}^n}\frac{|u(x)|^2}{|x-y|^{\gamma}}dx\leq C\|u\|^2_{\dot{H}^{\gamma/2}}
\end{equation}
\end{Lem}
\begin{proof} We denote the operator translation $\tau_y$ of $u$ by the vector $y\in\mathbb{R}^n$ as $\tau_yu(x):=u(x-y).$
Then $\tau_y$ is an isometry over space $\dot{H}^{\gamma/2}$, for $\gamma>0,$
\begin{equation}\label{iso}
\|\tau_yu\|^2_{\dot{H}^{\gamma/2}}
=\int_{\mathbb{R}^n}|\xi|^{\gamma}|\mathcal{F}(\tau_yu)(\xi)|^2d\xi
=\int_{\mathbb{R}^n}|\xi|^{\gamma}|\mathcal{F}(u)(\xi)|^2d\xi
=\|u\|^2_{\dot{H}^{\gamma/2}}.
\end{equation}
Since the Lebesgue  measure is invariant under translations and by  Remark \ref{rmk-3} part (ii) for $s:=\gamma/2$ together with the identity (\ref{iso}) we find that
\begin{equation*}
\begin{split}
\sup_{y\in \mathbb{R}^n}\int_{\mathbb{R}^n}\frac{|u(x)|^2}{|x-y|^{\gamma}}dx
&=\sup_{y\in \mathbb{R}^n}\int_{\mathbb{R}^n}\frac{|u(x-y)|^2}{|x|^{\gamma}}dx\\
&=\sup_{y\in \mathbb{R}^n}\left\|\frac{\tau_y u(\cdot)}{|\cdot|^{\gamma/2}}\right\|^2_{L^2}\\
&\leq C\sup_{y\in \mathbb{R}^n}\|\tau_y u\|^2_{\dot{H}^{\gamma/2}}\\
&= C\|u\|^2_{\dot{H}^{\gamma/2}},
\end{split}
\end{equation*}
for some  $C>0.$
\end{proof}

The next theorem is a direct consequence of {\cite[Theorem 6.3.2, page 148]{Bergh1976}}.
\begin{Thm}\label{teo-fund}Let $\beta>0$, $1<p<\infty,$ $u\in H^{\beta,p}.$ Then there exists a positive constant $C$ such that
\begin{equation}\label{Dot}
\|u\|_{\dot{H}^{\beta,p}}\leq C\|u\|_{H^{\beta,p}}.
\end{equation}
\end{Thm}
\begin{proof}
We denote  $m(\xi):=|\xi|^{\beta}(1+|\xi|^2)^{-\beta/2},$ then we notice that $m(\xi)$ is an $L^p-$Fourier multiplier on $\mathbb{R}^n$ for $1<p<\infty,$ see e.g., \cite[page 449]{Gr2008}. Now we let $u\in H^{\beta,p}$. Then it follows that
\begin{equation*}
\begin{split}
\|u\|_{\dot{H}^{\beta,p}}
&:=\|\mathcal{F}^{-1}[|\xi|^{\beta}\hat{u}(\xi)]\|_{L^p} \\
&=\|\mathcal{F}^{-1}[|\xi|^{\beta}(1+|\xi|^2)^{-\beta/2}(1+|\xi|^2)^{\beta/2}\hat{u}(\xi)]\|_{L^p}  \\
&=\|\mathcal{F}^{-1}[m(\xi)\mathcal{F}((I-\Delta)^{\beta/2}u)(\xi)]\|_{L^p} \\
&\leq C\|(I-\Delta)^{\beta/2}u\|_{L^p}=C\|u\|_{H^{\beta,p}},
\end{split}
\end{equation*}
where $C$ is a positive constant.
\end{proof}

\begin{Prop}\label{prop1} Let $\sigma>0,$ $1<r<\infty$ and $1<p_i,q_i\leq\infty$ and $\frac{1}{r}=\frac{1}{p_i}+\frac{1}{q_i},$ for $i=1,2.$
Suppose that $u\in H^{\sigma,p_1}\cap L^{p_2},v\in H^{\sigma,q_2}\cap L^{q_1}.$ Then $u\,v\in \dot{H}^{\sigma,r}$ and
\begin{equation}\label{Leib-G}
\displaystyle\|(-\Delta)^{\sigma/2}(uv)\|_{L^r}\leq C\left(\|u\|_{H^{\sigma,p_1}}\|v\|_{L^{q_1}}
 +\|u\|_{L^{p_2}}\|v\|_{H^{\sigma,q_2}}\right),
\end{equation}
where $C>0.$
\end{Prop}

\begin{proof}
Consider $D^{\sigma}:\mathcal{S}\times \mathcal{S}\to L^r$ the bilinear form defined as
\begin{equation*}
(u,v)\mapsto D^{\sigma}(u,v)=(-\Delta)^{\sigma/2}(uv).
\end{equation*}
We claim that $D^{\sigma}$ is continuous. In fact, we endow $\mathcal{S}\times \mathcal{S}$ with the norm of $(H^{\sigma,p_1}\cap L^{p_2})\times (H^{\sigma,q_2}\cap L^{q_1})$. Assume that
$(u_n,v_n)\in\mathcal{S}\times \mathcal{S}$ and $u_n\rightarrow0$ in the norm of $H^{\sigma,p_1}\cap L^{p_2},$ $v_n\rightarrow0$ in the norm of $H^{\sigma,q_2}\cap L^{q_1}$  that is
\begin{equation*}
\begin{split}
\|u_n\|_{H^{\sigma,p_1}\cap L^{p_2}}&=\|u_n\|_{H^{\sigma,p_1}}+\|u_n\|_{L^{p_2}}\xrightarrow[n\to \infty]{}0 \\
\|v_n\|_{H^{\sigma,q_2}\cap L^{q_1}}&=\|v_n\|_{H^{\sigma,q_2}}+\|v_n\|_{L^{q_1}}\xrightarrow[n\to \infty]{}0.
\end{split}
\end{equation*}
Hence by Remark \ref{rmk-3} part (iii) and Theorem \ref{teo-fund} we have that
\begin{equation*}
\begin{split}
\|D^{\sigma}(u_n,v_n)\|_{L^r}
&\leq \|u_n\|_{\dot{H}^{\sigma,p_1}}\|v_n\|_{L^{q_1}}+\|u_n\|_{L^{p_2}}\|v_n\|_{\dot{H}^{\sigma,q_2}}\\
&\leq \|u_n\|_{H^{\sigma,p_1}}\|v_n\|_{L^{q_1}}+\|u_n\|_{L^{p_2}}\|v_n\|_{H^{\sigma,q_2}}
\xrightarrow[n\to \infty]{}0,
\end{split}
\end{equation*}
Thus $D^{\sigma}$ is a continuous bilinear map on $\mathcal{S}\times\mathcal{S}.$ On the other hand, since
$\overline{\mathcal{S}}=H^{\sigma,p_1}\cap L^{p_2}$ and $\overline{\mathcal{S}}=H^{\sigma,q_2}\cap L^{q_1}$ (see \cite{Bergh1976}). Thus there exists a unique continuous extension of $D^{\sigma}$ from $H^{\sigma,p_1}\cap L^{p_2}\times H^{\sigma,q_2}\cap L^{q_1}\to L^r,$ in which we denote the extended map as $(-\Delta)^{\sigma/2}.$ Moreover by density it follows that the unique extension satisfy (\ref{Leib-G}), i.e,
\begin{equation*}
\displaystyle\|(-\Delta)^{\sigma/2}(uv)\|_{L^r}\leq C\left(\|u\|_{H^{\sigma,p_1}}\|v\|_{L^{q_1}}
 +\|u\|_{L^{p_2}}\|v\|_{H^{\sigma,q_2}}\right),
\end{equation*}
for each $u\in H^{\sigma,p_1}\cap L^{p_2},$ and $v \in H^{\sigma,q_2}\cap L^{q_1}.$
\end{proof}

\begin{Prop}\label{prop2} Let $\sigma>0,$ $1<r<\infty$ and $1<p_i,q_i\leq\infty$ and $\frac{1}{r}=\frac{1}{p_i}+\frac{1}{q_i},$ for $i=1,2.$
Suppose that $u\in H^{\sigma,p_1}\cap L^{p_2},$ and \,$v\in H^{\sigma,q_2}\cap L^{q_1}.$ Then $u\,v\in H^{\sigma,r}.$ Furthermore,
\begin{equation}\label{Leib1-15}
\displaystyle\|(I-\Delta)^{\sigma/2}(uv)\|_{L^r}\leq C\left(\|u\|_{H^{\sigma,p_1}}\|v\|_{L^{q_1}}+\|u\|_{L^{p_2}}\|v\|_{H^{\sigma,q_2}}\right),
\end{equation}
for some $C>0.$
\end{Prop}
\begin{proof} The proof is a direct consequence of the fractional Leibniz rule given by Remark \ref{rmk-3} (iii). Then,  as we argue in the proof of the Proposition \ref{prop1} the proof now follows.
\end{proof}

Next we state some properties for the operator $ K_{\gamma}$ defined in (\ref{htree}).


\begin{Lem}\label{lem1} Let $\gamma\in(0,n),$ $u\in \dot{H}^{\gamma/2}.$ Then there exists a positive constant $C$ such that
\begin{equation*}
\|K_{\gamma}(|u|^2)\|_{L^{\infty}}\leq C\|u\|^2_{\dot{H}^{\gamma/2}},
\end{equation*}
for $\psi\in L^{\infty}.$
\end{Lem}
\begin{proof}
We will consider $\gamma\in(0,n),$ $u\in \dot{H}^{\gamma/2}.$ Then by  Lemma \ref{des-gam} inequality (\ref{des-Hardy}) there exists $C>0$ such that
\begin{equation}\label{K-desig}
\begin{split}
\|K_{\gamma}(|u|^2)\|_{L^{\infty}}&:=\sup_{x\in \mathbb{R}^n}|(\psi_{\gamma}\ast |u|^2)(x)|\\
&\leq \|\psi\|_{L^{\infty}}\cdot \sup_{x\in \mathbb{R}^n}\int_{\mathbb{R}^n}\frac{|u(x-y)|^2}{|y|^{\gamma}}dy\\
&\leq C\|u\|^2_{\dot{H}^{\gamma/2}},
\end{split}
\end{equation}
where $\displaystyle\psi_{\gamma}=\frac{\psi}{|\cdot|^{\gamma}}, \psi\in L^{\infty}.$
\end{proof}


\begin{Lem}\label{lem2} Let $\gamma\in(0,n),$ $\beta\geq \gamma/2,$ $u\in H^{\beta}.$ Then there exists a positive constant $C$ such that
\begin{equation}
\|K_{\gamma}(|u|^2)\|_{L^{2n/\gamma}}\leq C\|u\|_{L^{2n/n-\gamma}}\|u\|_{H^{\beta}}.
\end{equation}
in which we assume that $\psi\in L^{\infty}.$
\end{Lem}
\begin{proof}
Let us take $u\in H^{\beta}.$ Then by inclusions (\ref{h1}), (\ref{h2}),  and H\"{o}lder´s inequality it follows that $|u|^2\in L^{2n/2n-\gamma}.$
Furthermore, from (\ref{htree}) it follows that
\begin{equation}\label{mod1}
\left|K_{\gamma}(|u|^2)(x)\right|\leq \|\psi\|_{\infty}\left(|u|^2\ast\frac{1}{|\cdot|^{\gamma}}\right)(x),\quad \psi\in L^{\infty}.
\end{equation}
Thus by Remark \ref{rmk-3} (v) for $q=2n/\gamma$, $p=2n/2n-\gamma$ and $|u|^2\in L^{2n/2n-\gamma}(\mathbb{R}^n)$ together with $H^{\beta}\hookrightarrow L^2$, there exists a constant $C>0$ such that,
\begin{equation}\label{K-2n}
\begin{split}
\|K_{\gamma}(|u|^2)\|_{L^{2n/\gamma}}
&\leq \|\psi\|_{\infty} \left\| |u|^2\ast\frac{1}{|\cdot|^{\gamma}}\right\|_{L^{2n/\gamma}}\\
&\leq C \|\psi\|_{\infty} \||u|^2\|_{L^{2n/2n-\gamma}}\\
&\leq C \|u\|_{L^{2n/n-\gamma}}\|u\|_{L^2}\\
&\leq C \|u\|_{L^{2n/n-\gamma}}\|u\|_{H^{\beta}},\quad u\in H^{\beta}.
\end{split}
\end{equation}
\end{proof}

\begin{Lem}\label{lem3} Let $\gamma\in(0,n)$, $\beta\geq\gamma/2$ be fixed. Then for each $u\in H^{\beta},$ there exists a positive constant $C$ such that
\begin{equation}\label{Kineq}
\|K_{\gamma}(|u|^2)\|_{H^{\beta,2n/\gamma}}\leq C\|u\|_{L^{2n/n-\gamma}}\|u\|_{H^{\beta}},
\end{equation}
for an arbitrary fixed $\psi\in L^{\infty}$ such that $|\psi(x)|\leq Me^{-\mu|x|},$ where $M,\mu\geq0, x\in\mathbb{R}^n.$
\end{Lem}

\begin{proof} We recall that $(I-\Delta)^{\beta/2}u={\mathcal{F}}^{-1}[(1+|\xi|^2)^{\beta/2}\widehat{u}(\xi)],$ belongs to $\mathcal{S}$\, for $u\in
{\mathcal{S}}.$
We claim that,
\begin{equation}\label{K1}
((I-\Delta)^{\beta/2}K_{\gamma})(|u|^2)=(K_{\gamma}(I-\Delta)^{\beta/2})(|u|^2),\,\,\,\text{for}\,\, u\in {\mathcal{S}}.
\end{equation}
Since,  $|\psi(x)|\leq Me^{-\mu|x|},$ decreases faster than any power of $1/|x|^{r}$ for each integers $r>0,$ by hypothesis. Then
$\psi_{\gamma}=\displaystyle\frac{\psi(\cdot)}{|\cdot|^{\gamma}}$ belongs to $\mathcal{O}_{C}'\subseteq \mathcal{S}',$ see e.g. \cite[Definition 30.1, page 315]{TR2016} for definition and properties of $\mathcal{O}_{C}'$. Moreover, if $u\in {\mathcal{S}}$ is a given function, then $|u|^2=u\overline{u}$ also defines an element of the Schwartz space. Therefore the convolution product
\begin{equation}\label{com}
K_{\gamma}(|u|^2)= (\psi_{\gamma}\ast|u|^2)
\end{equation}
exists as a tempered distribution. Hence,
\begin{equation}\label{com1}
{\cal{F}}(K_{\gamma}(|u|^2))=\mathcal{F}(\psi_{\gamma}\ast|u|^2)=
\mathcal{F}(\psi_{\gamma})\mathcal{F}(|u|^2).
\end{equation}
Furthermore, $(1+|\xi|^2)^{\beta/2}{\mathcal{F}}(|u|^2)\in \mathcal{S}$. Thus,

\begin{equation}\label{conmu}
\begin{split}
((I-\Delta)^{\beta/2}K_{\gamma})(|u|^2)&={\cal{F}}^{-1}[(1+|\xi|^2)^{\beta/2}{\cal{F}}(K_{\gamma}(|u|^2))] \\
&={\cal{F}}^{-1}\left[{\cal
{F}}(\psi_{\gamma})(1+|\xi|^2)^{\beta/2}{\cal{F}}(|u|^2)\right]\\
&=\left(\psi_{\gamma}\ast {\cal{F}}^{-1}[(1+|\xi|^2)^{\beta/2}{\cal{F}}(|u|^2)]\right)\\
&=(K_{\gamma}(I-\Delta)^{\beta/2})(|u|^2),\quad u\in \mathcal{S}.
\end{split}
\end{equation}
Since $\overline{\mathcal{S}}^{\|\cdot\|_{\beta}}=H^{\beta},$ $\beta>0.$ Then (\ref{conmu}) holds on $H^{\beta}.$ But then (\ref{conmu}) together with Remark \ref{rmk-3} (v) in which we choose $q=2n/\gamma,$ $p=2n/2n-\gamma$  implies that,
\begin{equation}\label{kgdes}
\begin{split}
\|K_{\gamma}(|u|^2)\|_{H^{\beta, 2n/\gamma}}&=\|(I-\Delta)^{\beta/2}K_{\gamma}(|u|^2))\|_{L^{2n/\gamma}}\\
& = \|K_{\gamma}(I-\Delta)^{\beta/2}(|u|^2))\|_{L^{2n/\gamma}}\\
&\leq C \|\psi\|_{\infty}\|(I-\Delta)^{\beta/2}(|u|^2)\|_{L^{2n/2n-\gamma}},\\
\end{split}
\end{equation}
Next, we estimate the right hand side of (\ref{kgdes}) by applying Fractional Leibniz rule (\ref{Leib1-15}). First we prove the following inequality
\begin{equation}\label{Kineq1}
\|(I-\Delta)^{\beta/2}(|u|^2)\|_{L^{2n/2n-\gamma}}\leq C\|u\|_{L^{2n/n-\gamma}}\|u\|_{H^{\beta}},
\end{equation}
In fact, from Proposition \ref{prop2}, in which  we now choose the parameters as follows:  $\sigma=\beta,$ $r=2n/2n-\gamma,$ $p_1=q_2=2,$  $q_1=p_2=2n/n-\gamma,$ and $u=v=|u|,$ where $u\in H^{\beta}\cap L^{2n/n-\gamma}.$ Then $|u|^2\in H^{\beta,2n/2n-\gamma}$ and
\begin{equation}\label{des-sob2}
\begin{split}
\|(I-\Delta)^{\beta/2}(|u|^2)\|_{L^{2n/2n-\gamma}}
&\leq C\big(\|(I-\Delta)^{\beta/2}|u|\|_{L^2}\|u\|_{L^{2n/n-\gamma}}+\|u\|_{L^{2n/n-\gamma}}\|(I-\Delta)^{\beta/2}|u|\|_{L^2}\big)\\
&=2C(\|u\|_{L^{2n/n-\gamma}}\|(I-\Delta)^{\beta/2}|u|\|_{L^2})\\
&=C(\|u\|_{L^{2n/n-\gamma}}\|u\|_{H^{\beta}}).\quad u\in H^{\beta}.
\end{split}
\end{equation}
Therefore from (\ref{kgdes}) and (\ref{des-sob2}) the proof follows.
\end{proof}

\begin{Lem}\label{lemx} Let $\gamma\in(0,n),$ $\beta\geq\gamma/2$. Then for each $u,v\in H^{\beta},$ there exists a positive constant $C$ such that
\begin{equation}\label{Estdif}
\||u|^2-|v|^2\|_{L^{2n/2n-\gamma}}\leq C(\|u\|_{H^{\beta}}+\|v\|_{H^{\beta}})\|u-v\|_{L^2}.
\end{equation}
\end{Lem}
\begin{proof} Let $u,$ $v$  be in $H^{\beta}.$ Then,
 $u, v\in L^{2n/n-\gamma},$ since $\beta\geq\gamma/2$. Moreover, $(|u|+|v|)|u-v|$ belongs to $L^{2n/2n-\gamma}.$ Thus, by H\"{o}lder's inequality and Sobolev inclusion,  it follows that
\begin{equation*}
\begin{split}
\||u|^2-|v|^2\|_{L^{2n/2n-\gamma}}
&\leq \|(|u|+|v|)|u-v|\|_{L^{2n/2n-\gamma}} \\
&\leq \||u|+|v|\|_{L^{2n/n-\gamma}}\|u-v\|_{L^2}\\
&\leq(\|u\|_{L^{2n/n-\gamma}}+\|v\|_{L^{2n/n-\gamma}})\|u-v\|_{L^2}\\
&\leq C(\|u\|_{H^{\beta}}+\|v\|_{H^{\beta}})\|u-v\|_{L^2},
\end{split}
\end{equation*}
where $C>0.$
\end{proof}


\begin{Lem}\label{lem4} For $\gamma\in(0,n),$ $\beta\geq\gamma/2,$ $n\geq1$ and $\psi\in L^{\infty}.$ Then there exists a positive constant $C$ such that the map $u\mapsto K_{\gamma}(|u|^2)u$ from $L^2$ to $L^2$ satisfies
\begin{equation*}
\|K_{\gamma}(|u|^2)u-K_{\gamma}(|v|^2)v\|_{L^2}\leq C(\|u\|^2_{H^{\beta}}+\|v\|^2_{H^{\beta}}+\|u\|_{H^{\beta}}\|v\|_{H^{\beta}})\|u-v\|_{L^2},
\end{equation*}
for $u,v\in H^{\beta}$.
\end{Lem}
\begin{proof} From the definition of $K_{\gamma}(|u|^2)$ given in (\ref{htree}) we have that,

\begin{align}\label{Our1}
\|K_{\gamma}(|u|^2)u-K_{\gamma}(|v|^2)v\|_{L^2}
&=\|K_{\gamma}(|u|^2)(u-v)+K_{\gamma}(|u|^2-|v|^2)v\|_{L^2}               \notag \\
&\leq\|K_{\gamma}(|u|^2)(u-v)\|_{L^2}+\|K_{\gamma}(|u|^2-|v|^2)v\|_{L^2}.
\end{align}
Next, applying  Lemma \ref{des-gam} for $\gamma\in(0,n)$ together with Theorem \ref{teo-fund} for $\beta=\gamma/2, p=2$ and $u\in H^{\gamma/2},$ we have that
\begin{equation*}
\begin{split}
\|K_{\gamma}(|u|^2)(u-v)\|^2_{L^2}
&=\int_{\mathbb{R}^n}|K_{\gamma}(|u|^2)(x)|^2|(u-v)(x)|^2dx \\
&\leq \sup_{x\in\mathbb{R}^n} |K_{\gamma}(|u|^2)(x)|^2\int_{\mathbb{R}^n}|(u-v)(x)|^2dx\\
&\leq C\|\psi\|_{L^{\infty}}^2\|u\|^4_{\dot{H}^{\gamma/2}}\int_{\mathbb{R}^n}|(u-v)(x)|^2dx\\
&\leq C \|u\|_{H^{\gamma/2}}^4\|u-v\|^2_{L^2}.
\end{split}
\end{equation*}
Thus from (\ref{h2}) we get for $\beta\geq \gamma/2$, that
\begin{equation}\label{do-1}
\|K_{\gamma}(|u|^2)(u-v)\|_{L^2}\leq C\|u\|^2_{H^{\gamma/2}}\|u-v\|_{L^2}\leq C\|u\|_{H^{\beta}}^2\|u-v\|_{L^2}.
\end{equation}
On the other hand, since $K_{\gamma}(|u|^2-|v|^2)\in L^{2n/\gamma},$ for $u, v\in H^{\beta}\hookrightarrow L^{2n/n-\gamma}$ by Lemma \ref{lem2} and the embedding (\ref{h2}). Moreover, we have that $|u|^2-|v|^2\in L^{2n/2n-\gamma}$ by Lemma \ref{lemx}.
Thus the second  summand on right hand side of (\ref{Our1}) satisfy

\begin{equation}\label{do-2}
\begin{split}
\|K_{\gamma}(|u|^2-|v|^2)v\|_{L^2}&\leq \|K_{\gamma}(|u|^2-|v|^2)\|_{L^{2n/\gamma}}\|v\|_{L^{2n/n-\gamma}} \\
&\leq \|v\|_{H^{\beta}}\|K_{\gamma}(|u|^2-|v|^2)\|_{L^{2n/\gamma}} \\
&\leq \|v\|_{H^{\beta}}\cdot\||u|^2-|v|^2\|_{L^{2n/2n-\gamma}}\\
&\leq \|v\|_{H^{\beta}}(\|u\|_{H^{\beta}}+\|v\|_{H^{\beta}})\|u-v\|_{L^2}.
\end{split}
\end{equation}
Thus, combining the estimate (\ref{Our1}), (\ref{do-1}), and (\ref{do-2}) we get that

\begin{equation*}
\begin{split}
\|K_{\gamma}(|u|^2)u-K_{\gamma}(|v|^2)v\|_{L^2}
&\leq \|K_{\gamma}(|u|^2)(u-v)\|_{L^2}+\|K_{\gamma}(|u|^2-|v|^2)v\|_{L^2}\\
&\leq C(\|u\|^2_{H^{\beta}}+\|v\|^2_{H^{\beta}}+\|u\|_{H^{\beta}}
\|v\|_{H^{\beta}})\|u-v\|_{L^2},
\end{split}
\end{equation*}
and the proof of Lemma \ref{lem4} is now complete.
\end{proof}

We next show that the nonlinear function $u\mapsto K_{\gamma}(|u|^2)u$ is Lipschitz continuous from the closed ball in $H^{\beta}$ into itself. To show this we state the following lemma.


\begin{Lem}\label{Lips0} For $\gamma/2<\beta<1,n\geq2$ and $\psi\in L^{\infty}.$ Then there exists a positive constant $C$ such that the map $u\mapsto K_{\gamma}(|u|^2)u,$ satisfies the following estimate on $H^{\beta},$
\begin{equation*}
\|K_{\gamma}(|u|^2)u-K_{\gamma}(|v|^2)v\|_{H^{\beta}}\leq C(\|u\|^2_{H^{\beta}}+\|u\|_{H^{\beta}}+\|v\|^2_{H^{\beta}}+\|u\|_{H^{\beta}}\|v\|_{H^{\beta}})\cdot\|u-v\|_{H^{\beta}},
\end{equation*}
for each $u,v\in H^{\beta}$.
\end{Lem}

\begin{proof} First we show that there exists a positive constant $C$ such that,
\begin{equation}\label{estim-B}
\|K_{\gamma}(|u|^2)u-K_{\gamma}(|v|^2)v\|_{\dot{H}^{\beta}}
\leq C\left(2\|u\|_{H^{\beta}}+2\|v\|_{H^{\beta}}(\|u\|_{H^{\beta}}+\|v\|_{H^{\beta}})\right)\cdot\|u-v\|_{H^{\beta}},
\end{equation}
for each $u,v\in H^{\beta}.$ Indeed, it follows from the definition of the convolution operator (\ref{htree}) that
\begin{equation}\label{prev}
K_{\gamma}(|u|^2)u-K_{\gamma}(|v|^2)v=K_{\gamma}(|u|^2)(u-v)+K_{\gamma}(|u|^2-|v|^2)v.
\end{equation}
Then proceed to estimate $K_{\gamma}(|u|^2)u-K_{\gamma}(|v|^2)v$ on the space $\dot{H}^{\beta}$ for $u,v\in H^{\beta}.$ In fact, from identity (\ref{prev}) we have that,
\begin{equation}\label{eqsep}
\begin{split}
\|(-\Delta)^{\beta/2}(K_{\gamma}(|u|^2)u-K_{\gamma}(|v|^2)v)\|_{L^2}
&\leq\|(-\Delta)^{\beta/2}K_{\gamma}(|u|^2)(u-v)\|_{L^2}\\
&+\|(-\Delta)^{\beta/2}K_{\gamma}(|u|^2-|v|^2)v\|_{L^2}.
\end{split}
\end{equation}
Thus, it is sufficient to obtain bounds for the following two quantities

\begin{equation*}
I:=\|(-\Delta)^{\beta/2}K_{\gamma}(|u|^2)(u-v)\|_{L^2},\ \  \mbox{and}\ \ J:=\|(-\Delta)^{\beta/2}K_{\gamma}(|u|^2-|v|^2)v\|_{L^2},\ u,v\in H^{\beta}.
\end{equation*}
For this purpose let us consider first the expression $I$. We notice that  $K_{\gamma}(|u|^2)$ belongs to $ H^{\beta,2n/\gamma}\cap L^{\infty},$ for $u\in H^{\beta},$ because of Theorem \ref{teo-fund}, Lemma \ref{lem1}, and Lemma \ref{lem3}. Moreover, if $u, v\in H^{\beta},$ then by the Proposition \ref{prop1} when  $\sigma=\beta,r=2,p_1=2n/\gamma,q_1=2n/n-\gamma,p_2=\infty,q_2=2,$ we obtain that
\begin{equation}\label{eqI}
\begin{split}
I&\leq \|K_{\gamma}(|u|^2)\|_{H^{\beta,2n/\gamma}}\|u-v\|_{L^{2n/n-\gamma}}
+\|K_{\gamma}(|u|^2)\|_{L^{\infty}}\|u-v\|_{H^{\beta}} \\
&=I_1+I_2.
\end{split}
\end{equation}
Now using the embedding (\ref{h2}) together with the Lemma \ref{lem3} we can estimate the first term of right side of (\ref{eqI}), that is,
\begin{equation}\label{desig-I1}
\begin{split}
I_1&=\|K_{\gamma}(|u|^2)\|_{H^{\beta,2n/\gamma}}\|u-v\|_{L^{2n/n-\gamma}}\\
&\leq C \|u\|_{L^{2n/n-\gamma}}\|u\|_{H^{\beta}}\|u-v\|_{H^{\beta}}\\
&\leq C \|u\|_{H^{\beta}}^2\|u-v\|_{H^{\beta}},\quad u,v\in H^{\beta}.
\end{split}
\end{equation}
Next, by the Lemma \ref{lem1}, Theorem \ref{teo-fund} for $p=2,$ $\beta=\gamma/2$ and the embedding (\ref{h2}), we can estimate the second term of (\ref{eqI}), that is,
\begin{equation}\label{desig-I2}
\begin{split}
I_2&=\|K_{\gamma}(|u|^2)\|_{L^{\infty}}\|u-v\|_{H^{\beta}}\\
&\leq C\|u\|_{\dot{H}^{\gamma/2}}^2\|u-v\|_{H^{\beta}} \\
&\leq C\|u\|_{H^{\gamma/2}}^2\|u-v\|_{H^{\beta}} \\
&\leq C\|u\|^2_{H^{\beta}}\|u-v\|_{H^{\beta}}, \quad u,v\in H^{\beta}.
\end{split}
\end{equation}

Next, it remains to obtain  estimates for $J=\|(-\Delta)^{\beta/2}K_{\gamma}(|u|^2-|v|^2)v\|_{L^2},$ for $u,v\in H^{\beta}.$ Once again we appeal to Proposition \ref{prop1} in the case that $\sigma=\beta,r=2,p_1=2n/\gamma,q_1=2n/n-\gamma,p_2=\infty,q_2=2.$ Hence we have that

\begin{equation}\label{eqJ}
\begin{split}
J&\leq\|K_{\gamma}(|u|^2-|v|^2)\|_{H^{\beta,2n/\gamma}}\|v\|_{L^{2n/n-\gamma}}
+\|K_{\gamma}(|u|^2-|v|^2)\|_{L^{\infty}}\|v\|_{H^{\beta}}\\
&=J_1+J_2.
\end{split}
\end{equation}
Thus by (\ref{h2}) and Remark \ref{rmk-3} (v) for $q=2n/\gamma,$ $p=2n/2n-\gamma$ together with Proposition \ref{prop2} for $\sigma=\beta,$ $r=2n/2n-\gamma,$ $p_1=2,q_1=2n/n-\gamma$ $p_2=2n/n-\gamma,q_2=2.$ Since $u,v,\in H^{\beta}\cap L^{2n/n-\gamma},$ it follows that

\begin{align}\label{estJ1}
J_1&=\|K_{\gamma}(|u|^2-|v|^2)\|_{H^{\beta,2n/\gamma}}\|v\|_{L^{2n/n-\gamma}}                     \notag \\
&\leq \|v\|_{H^{\beta}}\|K_{\gamma}(I-\Delta)^{\beta/2}(|u|^2-|v|^2)\|_{L^{2n/\gamma}}    \notag             \\
&\leq\|v\|_{H^{\beta}}\||u|^2-|v|^2\|_{H^{\beta,2n/2n-\gamma}}   \notag                                  \\
&\leq \|v\|_{H^{\beta}}\cdot\left(
\||u|+|v|\|_{H^{\beta}}\|u-v\|_{L^{2n/n-\gamma}}+\||u|+|v|\|_{L^{2n/n-\gamma}}\|u-v\|_{H^{\beta}}\right),
\end{align}

Since $H^{\beta}\hookrightarrow L^{2n/n-\gamma}$ we have

\begin{align}\label{est01}
\||u|+|v|\|_{H^{\beta}}\|u-v\|_{L^{2n/n-\gamma}}\leq(\|u\|_{H^{\beta}}+\|v\|_{H^{\beta}})\cdot\|u-v\|_{H^{\beta}},
\end{align}
and
\begin{equation}\label{est02}
\||u|+|v|\|_{L^{2n/n-\gamma}}\|u-v\|_{H^{\beta}}\leq(\|u\|_{H^{\beta}}+\|v\|_{H^{\beta}})\cdot\|u-v\|_{H^{\beta}}.
\end{equation}
Therefore, thanks to (\ref{estJ1})-(\ref{est02}) we obtain
\begin{equation}\label{desig-J1}
J_1\leq 2\|v\|_{H^{\beta}}(\|u\|_{H^{\beta}}+\|v\|_{H^{\beta}})\cdot\|u-v\|_{H^{\beta}},\ u,v\in H^{\beta}.
\end{equation}
Now, applying the same reasoning above we estimate $J_2.$ Thus,
\begin{equation}\label{desig-J2}
\begin{split}
J_2 &=\|K_{\gamma}(|u|^2-|v|^2)\|_{L^{\infty}}\|v\|_{H^{\beta}} \\
    &\leq C\|v\|_{H^{\beta}}(\|u\|_{H^{\beta}}+\|v\|_{H^{\beta}})\|u-v\|_{H^{\beta}}.
\end{split}
\end{equation}
Hence, the proof of assertion (\ref{estim-B}) follows from the estimates (\ref{eqI})-(\ref{desig-I2}), (\ref{eqJ}), (\ref{desig-J1})-(\ref{desig-J2}) together with the inequality (\ref{eqsep}).

On the other hand, due to the equivalence of norm provided by the Remark \ref{rmk-3} part (i) applied to $K_{\gamma}(|u|^2)u-K_{\gamma}(|v|^2)v$ under the condition $\beta\in (0,1),$ we have that
\begin{equation}\label{eqigua}
\begin{split}
\|K_{\gamma}(|u|^2)u-K_{\gamma}(|v|^2)v\|_{H^{\beta}}\leq\|K_{\gamma}(|u|^2)u-K_{\gamma}(|v|^2)v\|_{L^2} +\|(K_{\gamma}(|u|^2)u-K_{\gamma}(|v|^2)v)\|_{\dot{H}^{\beta}}.
\end{split}
\end{equation}
Thus, using the Lemma \ref{lem4} together with the estimate (\ref{estim-B}) and estimate (\ref{eqigua}) we get that
\begin{equation*}
\begin{split}
\|K_{\gamma}(|u|^2)u-K_{\gamma}(|v|^2)v\|_{H^{\beta}}
\leq C(\|u\|^2_{H^{\beta}}+\|u\|_{H^{\beta}}+\|v\|^2_{H^{\beta}}+\|u\|_{H^{\beta}}\|v\|_{H^{\beta}})\cdot\|u-v\|_{H^{\beta}},
\end{split}
\end{equation*}
for some positive constant $C$, and the proof of Lemma \ref{Lips0} is now complete.
\end{proof}

\section{ Non linear fractional Schr\"{o}dinger equation}\label{sec3}

In this section, we establish local existence in time for the fractional evolution problem
\begin{equation}\label{Hartree2}
\begin{cases}
i^{\alpha}D_t^{\alpha} u(t,x)=(-\Delta)^{\beta/2}u(t,x)+\lambda J_t^{1-\alpha} K_{\gamma}(|u|^2)(x)u(t,x), &\ \mbox{on}\ (0,T]\times\mathbb{R}^n \\
\quad\quad\ \ \ u(0,x)=u_0(x),\ u_0\in H^{\beta},     &
\end{cases}
\end{equation}
where $\alpha\in(0,1)$, $\beta>0$, $\lambda\in\mathbb{R}\backslash\{0\}.$ We consider in (\ref{Hartree2}), a Hartree type non-linearity, given by
\begin{equation*}
    K_{\gamma}(u)(x):=(|\cdot|^{-\gamma}\psi(\cdot)\ast u)(x),
\end{equation*}
where $\gamma\in(0,n),$ $\psi\in L^{\infty}.$  In this section we prove the existence and uniqueness of the solution for (\ref{Hartree2}). For this purpose, our main tool will be Banach's fixed point theorem and the results of the previous sections.

\subsection{Existence and local uniqueness}
In this section, we will prove the existence and uniqueness of solutions on  $C([0,T];H^{\beta})$ for equation (\ref{Hartree2}).

Hereafter we consider the $L^{\infty}$ norm on the space $C([0,T];H^{\beta}),$ that is,
\begin{equation}\label{notation}
\|u\|_{\infty}=\displaystyle\sup_{t\in [0,T]}\|u(t,\,\cdot)\|_{H^{\beta}}.
\end{equation}

We denote $u(t,\cdot)=u$ unless otherwise is specified. Furthermore, if $X$ is any of the function spaces under consideration, we simply write $u\in X$ whenever $u(t,\cdot)\in X$ for each $t\in [0,T].$

\begin{Def}[mild solution] Let $\alpha\in (0,1),$ $\beta>0$ be fixed. Assume that $u_0\in H^{\beta}.$ A function $u\in C([0,T];H^{\beta})$ is called a mild solution of (\ref{Hartree2}) if $u$ satisfies the integral equation
\begin{equation*}
\begin{split}
u(t,x)&=\mathcal{F}^{-1}[E_{\alpha}((-it)^{\alpha}|\xi|^{\beta})\hat{u}_0(\xi)](x) \\
&+\lambda(-i)^{\alpha}\int_0^t
\mathcal{F}^{-1}[E_{\alpha}((-i(t-s))^{\alpha}|\xi|^{\beta})\mathcal{F}(K_{\gamma}(|u|^2)u)(\xi)](x)ds
\end{split}
\end{equation*}
for each $t>0,x\in\mathbb{R}^n.$
\end{Def}

We are now ready to prove the main result of this paper.\\

\begin{Thm}\label{Thm1} Let $\alpha\in (0,1)$, $\gamma/2\leq\beta<1$, $n\geq2.$ Suppose that $u_0\in H^{\beta},$ $\psi\in L^{\infty}$ with  $|\psi(x)|\leq M\,e^{-\mu|x|},$ $M,\mu\geq0,x\in\mathbb{R}^n.$
Then, there exists $T>0,$ such that the nonlinear equation (\ref{Hartree2}) has  unique mild solution $u\in C([0,T];H^{\beta}),$ such that
\begin{equation}\label{desig-I}
    \|u\|_{\infty}\leq C\|u_0\|_{H^{\beta}},
\end{equation}
for some positive constant $C.$ Moreover, the map
\begin{equation*}
    \mathbb{F}:H^{\beta}\rightarrow C([0,T];H^{\beta}),\quad u_0\mapsto u(t,\cdot)
\end{equation*}
is continuous.
\end{Thm}

\begin{proof} Let us fix $T>0,$ and  choose  $ r<2\sqrt{2}M_0\|u_0\|_{H^{\beta}},$ in which the constant $M_0>0$ is taken from the Remark \ref{rmk-1}. We recall our notation (\ref{notation}),  that is, for a given $u\in C([0,T];H^{\beta})$ we have that $\|u(t)\|_{H^{\beta}}\leq\displaystyle\sup_{t\in [0,T]}\|u(t)\|_{H^{\beta}}=\|u\|_{\infty},$ for  all $t$. Furthermore, we denote the closed ball of radius $r$ on $C([0,T];H^{\beta})$ as
\begin{equation*}
    B_r=\left\{u\in C([0,T];H^{\beta}):\|u\|_{\infty}\leq r\right\}.
\end{equation*}

 Next, under these considerations we define the nonlinear operator $\Phi_{u_0}:B_r\to B_r$  by

\begin{equation}\label{Trans1}
\begin{split}
\Phi_{u_0}(u)(t,x)&=\mathcal{F}^{-1}[E_{\alpha}((-it)^{\alpha}|\xi|^{\beta})\hat{u}_0(\xi)](x) \\
&+\lambda(-i)^{\alpha}\int_0^t
\mathcal{F}^{-1}[E_{\alpha}((-i(t-s))^{\alpha}|\xi|^{\beta})\mathcal{F}(K_{\gamma}(|u|^2)u)(\xi)](x)ds
\end{split}
\end{equation}
First we claim that $\Phi_{u_0}$ is well defined and $\Phi_{u_0}$ maps $B_r$ to $B_r.$ We notice that $t\mapsto\Phi_{u_0}u(t)$ is continuous. Moreover,
since the mapping $\xi\mapsto E_{\alpha}((-it)^{\alpha}|\xi|^{\beta})$
is bounded by Remark \ref{rmk-1}. Then for $u_0\in H^{\beta}\subseteq L^2$, and $0<\gamma\leq 2\beta.$ It follows  by H\"{o}lder's inequality, that

\begin{align}\label{TransO}
|\mathcal{F}\Phi_{u_0}(u)(t,\xi)|^2
&\leq 2|E_{\alpha}((-it)^{\alpha}|\xi|^{\beta})\hat{u}_0(\xi)|^2 \notag\\
&+2\lambda^2\left|\int_0^tE_{\alpha}((-i(t-s))^{\alpha}|\xi|^{\beta})\mathcal{F}(K_{\gamma}(|u|^2)u)(\xi)ds\right|^2 \notag\\
&\leq 2M_0^2|\hat{u}_0(\xi)|^2 \notag\\
&+2\lambda^2 \int_0^t |E_{\alpha}((-i(t-s))^{\alpha}|\xi|^{\beta})|^2ds\cdot \int_0^t|\mathcal{F}(K_{\gamma}(|u|^2)u)(\xi)|^2ds \notag\\
&\leq 2M_0^2|\hat{u}_0(\xi)|^2+2T\lambda^2 M_0^2\int_0^t|\mathcal{F}(K_{\gamma}(|u|^2)u)(\xi)|^2ds,\quad t>0,\xi\in\mathbb{R}^n,
\end{align}

Thus,  from (\ref{TransO}) we obtain that

\begin{equation}\label{S-norm}
\begin{split}
\int_{\mathbb{R}^n}(1+|\xi|^2)^{\beta}|\mathcal{F}\Phi_{u_0}(u)(t,\xi)|^2d\xi
&\leq 2M_0^2\int_{\mathbb{R}^n}(1+|\xi|^2)^{\beta}|\hat{u}_0(\xi)|^2d\xi\\
&+2T\lambda^2 M_0^2\int_0^t\int_{\mathbb{R}^n}(1+|\xi|^2)^{\beta}|\mathcal{F}(K_{\gamma}(|u|^2)u(\xi)|^2d\xi ds \\
&=2M_0^2\|u_0\|_{H^{\beta}}^2+2T\lambda^2 M_0^2\int_0^t\|K_{\gamma}(|u|^2)u\|^2_{H^{\beta}}ds,
\end{split}
\end{equation}

for every $u_0\in H^{\beta}.$ Then, it follows from (\ref{S-norm}) that
\begin{equation}\label{Phi-u}
\begin{split}
\|\Phi_{u_0}(u)\|_{\infty}
&=\sup_{t\in[0,T]}\left(\int_{\mathbb{R}^n}(1+|\xi|^2)^{\beta}|\mathcal{F}\Phi_{u_0}(u)(t,\xi)|^2d\xi\right)^{1/2} \\
&\leq\sqrt{2} M_0\|u_0\|_{H^{\beta}}+\sqrt{2}T\lambda M_0\|K_{\gamma}(|u|^2)u\|_{\infty}.
\end{split}
\end{equation}
Therefore, it suffices  to estimate $\|K_{\gamma}(|u|^2)u\|_{H^{\beta}}$ to ensure that $\Phi_{u_0}(u)\in B_r$.
Thus from the Remark \ref{rmk-3} part (i) for $\beta\in (0,1),$ we have that
\begin{equation}\label{exp-equiv}
\|K_{\gamma}(|u|^2)u\|_{H^{\beta}}\leq \|K_{\gamma}(|u|^2)u\|_{L^2}+\|(-\Delta)^{\beta/2}K_{\gamma}(|u|^2)u\|_{L^2},\quad u\in H^{\beta}.
\end{equation}
In order to estimate both quantities of the right hand side of (\ref{exp-equiv}). We notice that  $K_{\gamma}(|u|^2)\in L^{\infty}$ by Lemma \ref{lem1}. But then $K_{\gamma}(|u|^2)u\in L^2$ and
\begin{equation}\label{des-1}
\|K_{\gamma}(|u|^2)u\|_{L^2}\leq\|K_{\gamma}(|u|^2)\|_{L^{\infty}}\|u\|_{L^2}.
\end{equation}

On the other hand, since $u\in \dot{H}^{\beta}\cap L^{2n/n-\gamma}$ then $K_{\gamma}(|u|^2)\in \dot{H}^{\beta,2n/\gamma}\cap L^{\infty}$,   by Lemma \ref{lem1} together with Lemma \ref{lem3}.  Furthermore let us assume in Proposition \ref{prop1} that  $\sigma=\beta,$ $r=2,p_1=2n/\gamma,q_1=2n/n-\gamma,$ $p_2=\infty,$ $q_2=2.$ But then follows that $K_{\gamma}(|u|^2)u\in \dot{H}^{\beta}.$


Hence,
\begin{equation}\label{des-2}
\begin{split}
\|(-\Delta)^{\beta/2}K_{\gamma}(|u|^2)u\|_{L^2}
&\leq C(\|K_{\gamma}(|u|^2)\|_{H^{\beta,2n/\gamma}}\|u\|_{L^{2n/n-\gamma}}+\|K_{\gamma}(|u|^2)\|_{L^{\infty}}\|u\|_{H^{\beta}}).
\end{split}
\end{equation}

Thus, by (\ref{exp-equiv}) together with (\ref{des-1}),(\ref{des-2}) we obtain that
\begin{equation}\label{K-u}
\begin{split}
\|K_{\gamma}(|u|^2)u\|_{H^{\beta}}
&\leq C(\|K_{\gamma}(|u|^2)\|_{L^{\infty}}\cdot\left(\|u\|_{L^2}+\|u\|_{H^{\beta}}\right)+\|K_{\gamma}(|u|^2)\|_{H^{\beta,2n/\gamma}}\|u\|_{L^{2n/n-\gamma}})\\
&\leq C(\|K_{\gamma}(|u|^2)\|_{L^{\infty}}\|u\|_{H^{\beta}}+
\|K_{\gamma}(|u|^2)\|_{H^{\beta,2n/\gamma}}\|u\|_{L^{2n/n-\gamma}}),
\end{split}
\end{equation}
for some constant $C.$

Therefore, because of (\ref{h2}), (\ref{K-u}), Lemma \ref{lem1}, and Lemma \ref{lem3} we obtain that
\begin{equation}\label{K-3}
\begin{split}
\|K_{\gamma}(|u|^2)u\|_{H^{\beta}}
&\leq C(\|u\|^2_{\dot{H}^{\gamma/2}}\|u\|_{H^{\beta}}+\|u\|_{H^{\beta}}\|u\|^2_{L^{2n/n-\gamma}})   \\
&\leq C\|u\|^3_{H^{\beta}}.
\end{split}
\end{equation}

Hence, from (\ref{Phi-u}) and (\ref{K-3}) we obtain,
\begin{equation*}
\begin{split}
\|\Phi_{u_0}(u)\|_{\infty}
&\leq \sqrt{2}M\|u_0\|_{H^{\beta}}+3\sqrt{2}C\sqrt{T}\lambda M_0\|u\|^3_{\infty} \\
&\leq\frac{r}{2}+3\sqrt{2}T\lambda C M_0r^3\leq r,
\end{split}
\end{equation*}
if $T>0$ is small enough, we can conclude that the operator $\Phi_{u_0}$ leaves the closed ball $B_r$ invariant.

Next, we show that $\Phi_{u_0}$ is an operator Lipschitz for $T$ sufficiently small. In what follows we assume that $u,v$ belongs to $B_r\subseteq C([0,T];H^{\beta}),$ then we have that
\begin{equation*}
\|\Phi_{u_0}(u)-\Phi_{u_0}(v)\|_{\infty}=\sup_{t\in [0,T]}\|(\Phi_{u_0}(u)-\Phi_{u_0}(v))(t,\cdot)\|_{H^{\beta}},
\end{equation*}
in which we denote
\begin{equation*}
\begin{split}
\Gamma_{u_0}(u,v)(t)&:=(\Phi_{u_0}(u)-\Phi_{u_0}(v))(t) \\
&=\lambda(-i)^{\alpha}\int_0^t \mathcal{F}^{-1}\left[E_{\alpha}((-i(t-s))^{\alpha}|\xi|^{\beta})
[\mathcal{F}(K_{\gamma}(|u|^2)u)(\xi)
-\mathcal{F}(K_{\gamma}(|v|^2)v)(\xi)]\right](x)ds.
\end{split}
\end{equation*}
Then, according to Remark \ref{rmk-1}, H\"{o}lder inequality and the last equality we get that
\begin{equation*}
\begin{split}
\|\Gamma_{u_0}(u,v)\|_{H^{\beta}}^2
&=\int_{\mathbb{R}^n}|\mathcal{F}(\Phi_{u_0}(u)-\Phi_{u_0}(v))(t,\xi)|^2d\mu(\xi),\quad d\mu(\xi):=(1+|\xi|^2)^{\beta}d\xi\\
&=\lambda^2\int_{\mathbb{R}^n}\left|\int_0^t E_{\alpha}((-i(t-s))^{\alpha}|\xi|^{\beta})
\cdot [\mathcal{F}(K_{\gamma}(|u|^2)u)(\xi)-\mathcal{F}(K_{\gamma}(|v|^2)v)(\xi)]ds\right|^2d\mu(\xi) \\
&\leq \lambda^2\int_{\mathbb{R}^n}\left(\int_0^t|E_{\alpha}((-i(t-s))^{\alpha}|\xi|^{\beta})|^2ds \right. \\
& \cdot\left.\int_0^t|\mathcal{F}(K_{\gamma}(|u|^2)u)(\xi)-\mathcal{F}(K_{\gamma}(|v|^2)v)(\xi)|^2ds \right)d\mu(\xi) \\
&\leq \lambda^2 M_0^2 T\int_0^T\int_{\mathbb{R}^n}(1+|\xi|^2)^{\beta}|\mathcal{F}(K_{\gamma}(|u|^2)u)(\xi)-\mathcal{F}(K_{\gamma}(|v|^2)v)(\xi)|^2d\xi ds\\
&\leq (\lambda T M_0)^2\|K_{\gamma}(|u|^2)u-K_{\gamma}(|v|^2)v\|_{\infty}^2,
\end{split}
\end{equation*}
that is, $\|\Phi_{u_0}(u)-\Phi_{u_0}(v)\|_{\infty}\leq \lambda TM_0\|K_{\gamma}(|u|^2)u-K_{\gamma}(|v|^2)v\|_{\infty},$ for $\lambda\neq0,T>0, M_0$ positive constant. In fact, by the Lemma \ref{Lips0} we find the following estimate

\begin{equation*}
\begin{split}
\|K_{\gamma}(|u|^2)u-K_{\gamma}(|v|^2)v\|_{H^{\beta}}
&\leq C(\|u\|^2_{H^{\beta}}+\|u\|_{H^{\beta}}+\|v\|^2_{H^{\beta}}+\|u\|_{H^{\beta}}\|v\|_{H^{\beta}})\cdot\|u-v\|_{H^{\beta}}\\
&\leq C(3r^2+r)\cdot\|u-v\|_{H^{\beta}},
\end{split}
\end{equation*}
for some constant $C$. In this way we have that
\begin{equation*}
\begin{split}
\|\Phi_{u_0}(u)-\Phi_{u_0}(v)\|_{\infty}
&=\sup_{t\in [0,T]}\|(\Phi_{u_0}(u)-\Phi_{u_0}(v))(t,\cdot)\|_{H^{\beta}} \\
&\leq (3r^2+r)T\lambda CM_0\|u-v\|_{\infty},
\end{split}
\end{equation*}
then, if we assume $(3r^2+r)T\lambda CM_0<1$, we get that $\Phi_{u_0}$ defines a contraction on closed ball $B_r$.

It remains to prove the continuous dependence of $\Phi(u(t))=\Phi_{u_0}(u(t))$ with respect to $u_0,$ we notice that if $u,v$ are the corresponding mild solutions of (\ref{Hartree2}) with initial data $u_0,v_0$, respectively. Thus, we have that

\begin{align}\label{cont0}
u(t,x)-v(t,x)&=\mathcal{F}^{-1}[E_{\alpha}((-it)^{\alpha}|\xi|^{\beta})(\hat{u}_0(\xi)-\hat{v}_0(\xi))](x)+\lambda(-i)^{\alpha}\int_0^t \mathcal{F}^{-1}[E_{\alpha}((-i(t-s))^{\alpha}|\xi|^{\beta}) \notag\\
&\cdot(\mathcal{F}(K_{\gamma}(|u(s,\cdot)|^2)u(s,\cdot))(\xi)-\mathcal{F}(K_{\gamma}(|v(s,\cdot)|^2)v(s,\cdot))(\xi))](x)ds.
\end{align}
From (\ref{cont0}) we have that
\begin{align}\label{cont1}
\widehat{u(t,\cdot)}(\xi)-\widehat{v(t,\cdot)}(\xi)&=E_{\alpha}((-it)^{\alpha}|\xi|^{\beta})(\hat{u}_0(\xi)-\hat{v}_0(\xi))+\lambda(-i)^{\alpha}\int_0^t
E_{\alpha}((-i(t-s))^{\alpha}|\xi|^{\beta})   \notag\\
&\cdot [\mathcal{F}(K_{\gamma}(|u(s,\cdot)|^2)u(s,\cdot))(\xi)-\mathcal{F}(K_{\gamma}(|v(s,\cdot)|^2)v(s,\cdot))(\xi)]ds.
\end{align}
Therefore from (\ref{cont1}) together with the fact that
\begin{equation*}
\|K_{\gamma}(|u(s,\cdot)|^2)u(s,\cdot)-K_{\gamma}(|v(s,\cdot)|^2)v(s,\cdot)\|_{H^{\beta}}
\leq C\|u(s,\cdot)-v(s,\cdot)\|_{H^{\beta}},\quad u,v\in B_r\subseteq H^{\beta},
\end{equation*}
by Lemma \ref{Lips0}.
Next we denote by,
\begin{equation*}
A=2\int_{\mathbb{R}^n}|E_{\alpha}((-it)^{\alpha}|\xi|^{\beta})|^2|\hat{u}_0(\xi)-\hat{v}_0(\xi)|^2d\mu(\xi),
\end{equation*}
 and
\begin{equation*}
B=2\lambda^2\int_{\mathbb{R}^n}\left|\int_0^tE_{\alpha}((-i(t-s))^{\alpha}|\xi|^{\beta})\cdot [\mathcal{F}(K_{\gamma}(|u(s)|^2)u(s))-\mathcal{F}(K_{\gamma}(|v(s)|^2)v(s))]ds \right|^2d\mu(\xi),
\end{equation*}
Then we obtain that
\begin{equation}\label{est-uv}
\begin{split}
\|u(t,\cdot)-v(t,\cdot)\|^2_{H^{\beta}}
&=\int_{\mathbb{R}^n}(1+|\xi|^2)^{\beta}|\hat{u}(t,\xi)-\hat{v}(t,\xi)|^2d\xi,\\
&\leq A+B.\\
\end{split}
\end{equation}
Since,
\begin{equation}
A\leq 2M_0^2\|u_0-v_0\|^2_{H^{\beta}}
\end{equation}
and
\begin{equation}
B\leq 2T\lambda^2M_0^2\int_0^T\|K_{\gamma}(|u(s)|^2)u(s)-K_{\gamma}(|v(s)|^2)v(s)\|^2_{H^{\beta}} ds
\leq2(T\lambda CM_0)^2\|u-v\|^2_{\infty},
\end{equation}
it then follows that
\begin{equation}\label{est-uv}
\|u(t,\cdot)-v(t,\cdot)\|^2_{H^{\beta}}
\leq 2M_0^2\|u_0-v_0\|^2_{H^{\beta}}+2(T\lambda CM_0)^2\|u-v\|^2_{\infty},
\end{equation}
for each $u,v\in B_r\subseteq H^{\beta}.$

Hence by (\ref{est-uv}) it follows
\begin{equation}\label{est-uv1}
\begin{split}
\|u(t,\cdot)-v(t,\cdot)\|_{H^{\beta}}&\leq \left( 2M_0^2\|u_0-v_0\|^2_{H^{\beta}}+2(T\lambda CM_0)^2\|u-v\|^2_{\infty}\right)^{1/2}\\
                   &\leq \sqrt{2}M_0\|u_0-v_0\|_{H^{\beta}}+\sqrt{2}T\lambda CM_0\|u-v\|_{\infty}.
\end{split}
\end{equation}
Hence, by taking the supremum in $t$  on the left hand side of (\ref{est-uv1}) we obtain that
\begin{equation*}
\|u-v\|_{\infty}\leq \frac{\sqrt{2}M_0}{1-\sqrt{2}T\lambda CM_0}\|u_0-v_0\|_{H^{\beta}},
\end{equation*}
in which $T$ is appropriate.

%
\end{proof}
\paragraph{Acknowledgements:}
This work  has been partially supported by FONDECYT grant  \# 1170571.

\end{document}